\documentclass[10pt,conference,a4paper]{IEEEtran}
\usepackage[left=1.60cm, right=1.60cm, top=1.6cm, bottom=1.6cm]{geometry}
\usepackage[fleqn]{amsmath}
\usepackage[pdftex]{graphicx}
\usepackage{amssymb}
\usepackage{amsthm}
\usepackage{bm}
\usepackage{color}
\usepackage{epstopdf}
\newtheorem{thm}{Theorem}

\newtheorem{defn}{Definition}

\newtheorem{lem}{Lemma}
\newtheorem{rem}{Remark}
\graphicspath{ {Figures/} }
\title{ Information and Power Transfer by Energy Harvesting Transmitters over a Fading Multiple Access Channel with Minimum Rate Constraints }

\author{Deekshith P K \IEEEauthorrefmark{1}, Trupthi Chougule\IEEEauthorrefmark{2}, Shreya Turmari\IEEEauthorrefmark{2}, Ramya Raju\IEEEauthorrefmark{2}, Rakshitha Ram\IEEEauthorrefmark{2} and Vinod Sharma\IEEEauthorrefmark{1} \\ {\IEEEauthorrefmark{1}ECE Dept., Indian Institute of Science, Bangalore, India. \IEEEauthorrefmark{2}PES Institute of Technology, Bangalore, India.}\\{\hspace{-.6cm} \IEEEauthorrefmark{1}Email: \{deeks, vinod\}@ece.iisc.ernet.in},\\ \IEEEauthorrefmark{2}Email:\{chouguletrupthi, shreya.turmari, ramyaraju93, rakshitharam.93\}@gmail.com.}

\IEEEoverridecommandlockouts
\begin{document}
\maketitle
\begin{abstract}
We consider the problem of Simultaneous Wireless Information and Power Transfer (SWIPT) over a fading multiple access channel with additive Gaussian noise. The transmitters as well as the receiver harvest energy from ambient sources. We assume that the transmitters have two classes of data to send, viz. delay sensitive and delay tolerant data. Each transmitter sends the delay sensitive data at a certain minimum rate irrespective of the channel conditions (fading states). In addition, if the channel conditions are good, the delay tolerant data is sent.  {Along with data, the transmitters also transfer power to aid the receiver in meeting its energy requirements.} In this setting, we characterize the \textit{minimum-rate capacity region} which provides the fundamental limit of transferring information and power simultaneously with minimum rate guarantees. Owing to the limitations of current technology, these limits might not be achievable in practice. Among the practical receiver structures proposed for SWIPT in literature, two popular architectures are the \textit{time switching} and \textit{power splitting} receivers. For each of these architectures, we derive the  minimum-rate capacity regions. We show that power splitting receivers although more complex, provide a larger capacity region.
\end{abstract}
\noindent
\section{Introduction}

The abstract idea of transmitting power over a noisy communication link using an information carrying symbol  \cite{varshney2008transporting} finds   appealing applications in the context of energy deprived low power sensor networks,  Internet of Things (IoT) and potentially in many next generation communication systems. Of particular interest, is the idea of Simultaneous Wireless Information and Power Transfer (SWIPT). The topic has received considerable research attention in recent years (\cite{lu2015wireless}).
\par The fundamental trade-off in jointly transferring information and energy over a discrete memoryless channel was first studied in \cite{varshney2008transporting}. This work characterized the inherent trade-off in terms of a capacity-energy function. The idea was further pursued in \cite{grover2010shannon} and the authors obtained an optimal trade-off  between the achievable rate versus  transferred power over a point-to-point frequency selective channel. One important concern in designing simultaneous information and power harnessing systems is the receiver architecture. Two practical receiver models for SWIPT are the \textit{time-switching receiver} and the \textit{power-splitting receiver} (\cite{zhang2013mimo}). A dynamic power splitting approach is pursued in \cite{zhou2013wireless}. 
 \par In the case of multi-user channels, \cite{fouladgar2012transfer} provides the capacity-energy region of a discrete memoryless multiple access channel (MAC) and the capacity-energy function of a discrete memoryless multi-hop channel. Recent work \cite{amor2015feedback} considers a  Gaussian MAC (GMAC) with and without feedback, derives the capacity-energy region and proves that feedback can strictly increase the capacity-energy region of the channel. In their model in \cite{amor2015feedback}, the authors allow cooperation among transmitters. In \cite{hadzi2014multiple}, a model wherein a base station powers mobile devices in the downlink and the devices transfer data in the uplink is considered. Achievable rate regions are obtained for this case.
\par In this work, we address the problem of characterizing fundamental trade-off in jointly transferring information and power over a fading GMAC, simultaneously ensuring a certain quality of service (QoS) requirement. The QoS parameter we are interested in, is one of minimum instantaneous rate. To the best our knowledge, this is the first work characterizing the fundamental limits of communication of an SWIPT GMAC system with minimum rate guarantees.  Also in this work, unlike \cite{fouladgar2012transfer}, \cite{amor2015feedback} transmitters harvest energy from their ambient sources and the receiver also has an ambient energy harvesting source (e.g. sun or wind). We model energy consumption at the receiver explicitly by a discrete time stochastic process taking non-negative values. The realizations of this process are assumed to be \emph{not} known to the transmitters. All the fundamental rate regions that we derive in this work are achievable under this assumption. 

Owing to the limitations in technology, capacity-energy function \cite{varshney2008transporting} provides upper bound for rates achievable in current SWIPT systems. An \textit{ideal receiver} is one which harvests energy from a data stream without \textit{distorting} its information content and hence can achieve this upper bound. First, we derive the minimum-rate capacity region of a fading GMAC with energy harvesting transmitters and an ideal receiver. Next, we characterize the minimum-rate capacity region when the receivers are  time-switching and power-splitting, respectively. 

Following is the organization of the paper. In Section \ref{S_Prel}, we present the system model and introduce notation. In Section \ref{S_ID}, we derive the minimum-rate capacity region for the model with the  ideal receiver. In Section \ref{S_TS}, we obtain the minimum-rate erasure capacity region for the time-switching receiver model.  In Section \ref{S_PS}, we obtain the minimum-rate capacity region for the power-splitting receiver model. In Section \ref{S_NR}, we provide numerical results. We conclude the paper in Section \ref{S_Conc}.

\section{System Model}

\label{S_Prel}
\begin{figure}[h]
\begin{center}
\includegraphics[scale=0.45]{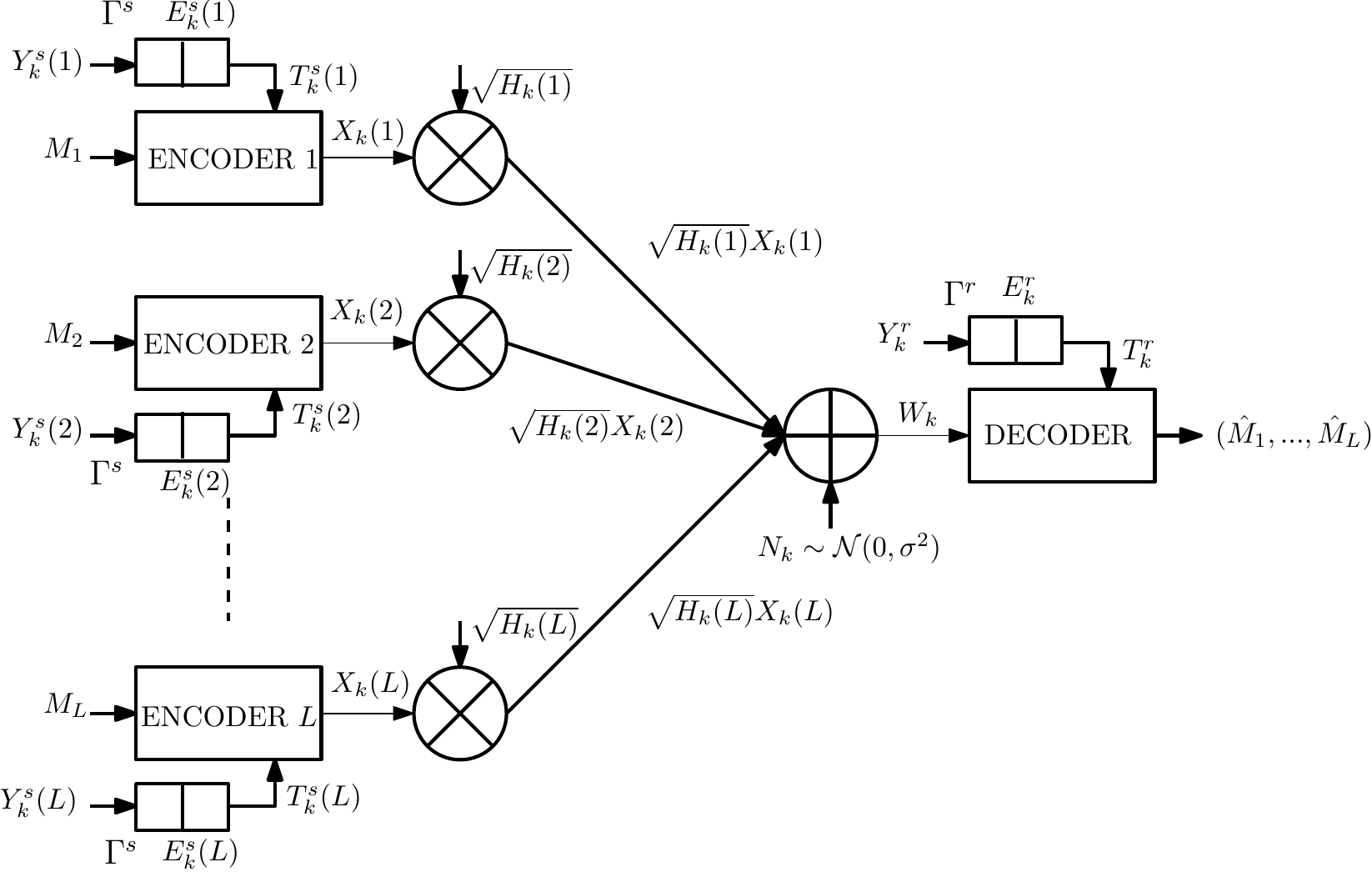}
\caption{GMAC system model} \label{fig1}
\end{center}
\end{figure}

Our system model is motivated by a multihop sensor network where all the sensor nodes are harvesting energy from the environment. Our model of a MAC is a basic building block in such a network. The receiver node is also a sensor node which collects information from the uplink nodes to transmit further to a downlink fusion node or another sensor node for further transmission. Since the receiver node may have to transmit its own data as well as of the uplink nodes, it needs more energy and hence may receive some energy from the uplink nodes along with data. 

\subsection{Transmitter and Channel Model}
 We consider $L$  energy harvesting transmitters transferring information  and power simultaneously to a spatially separated receiver (Figure \ref{fig1}). The transmitters and the receiver are equipped with a battery (energy buffer or simply, buffer) of  infinite storage capacity. We consider a slotted system. In slot $k$, transmitter $i$ harvests $Y_k^s(i)\in \mathcal{Y}^s \subseteq \mathbb{R}^+$ amount of energy wherein $s$ stands for sender and $\mathbb{R}^+$ denotes the positive real axis. For each $i \in \{1,2,\hdots L\}\equiv [1:L]$, $\{Y_k^s(i),~ k \geq 1\}$ is a stationary, ergodic process independent of each other. At transmitter $i$, $\mathbb{E}[Y^s(i)]$  denotes the mean energy harvested per slot. 

Transmitter $i$ selects a message $M_i$, uniformly randomly from its message set $\mathcal{M}_i$ with cardinality $|\mathcal{M}_i|$. This message is communicated to the receiver using an $n$ length codeword, denoted as $X'^n(i)$. Let  $X_k(i) \in \mathcal{X}$ denote the channel input symbol (possibly different from $X_k'(i)$ due to energy harvesting constraints) of transmitter $i$ in slot $k$. Accordingly, $\{X_k(M_i),~ k\geq 1\}$ denotes the channel input process corresponding to the message $M_i$. 

The messages are communicated across a memoryless Gaussian Multiple Access Channel (GMAC) subjected to a flat fading process. Let $H_k(i)$ denote the channel gain in the $k^{\text{th}}$ slot, of the channel between the transmitter $i$ and the receiver. For each $i$, let $\{H_k(i),~k\geq 1 \}$ be a stationary, ergodic process with { joint stationary distribution } $F_\mathbf{H}(.)$. We assume that the instantaneous realizations of $H_k(i)\in \mathcal{H} \subseteq \mathbb{R}^+$, referred as Channel State Information (CSI), are known at the transmitters and the receiver. The fading process is assumed to be independent across the users. In slot $k$, the channel output is  $W_k=\sum\limits_{i=1}^{L} \sqrt{H_k(i)}X_k(i)+N_k,$ where $\{N_k,~ k \geq 1\}$ is an i.i.d Gaussian process with mean $0$, variance $\sigma^2$ and the probability density function of $N_k$ is denoted by  $\mathcal{N}(0,\sigma^2)$.

 In slot $k$, transmitter $i$ adopts an energy management policy $T_k^s(i):\mathcal{M}_i \times \mathcal{Y}_s^k \times \mathcal{H}^L  \rightarrow \mathcal{T}^s(i) \subseteq \mathbb{R}^+$ which is a stochastic mapping where the Cartesian products $\mathcal{Y}_s^k \triangleq \mathcal{Y}^s \times \hdots \times \mathcal{Y}^s $ ($k$ times), $\mathcal{H}^L \triangleq \mathcal{H} \times \hdots \times \mathcal{H} $ ($L$ times). Policy $T_k^s(i)$ determines the amount of energy transmitter $i$ spends in slot $k$. Let $E_k^s(i)\in \mathcal{A}^s \subseteq \mathbb{R}^+$ denote the available energy in the $i^{\text{th}}$ transmitter's  buffer at the beginning of slot $k$. We assume that the harvested  energy is available for transmission in the same slot. Thus, total energy available for transmission in the  $k^\text{th}$ slot at transmitter $i$ is $E_k^s(i)+Y_k^s(i) \equiv \hat{E}_k^s(i) \in \hat{\mathcal{A}}^s \subseteq \mathbb{R}^+$, i.e., $T_k^s(i) \leq \hat{E}_k^s(i)$. The evolution of the process $\{E_k^s(i),~k \geq 1\}$ is as follows: $E_{k+1}^s(i)= \hat{E}_k^s(i)-T_k^s(i) .$

\subsection{Receiver Model}

For an SWIPT system, the receiver module typically consists of a rectenna unit and a communication module \cite{krikidis2014simultaneous}. We model the total energy consumption at the receiver by a stationary, ergodic process $\{T_k^r,~ k \geq 1\}$ (here $r$ denotes receiver). As noted in \cite{mahdavi2013energy}, dominant energy consumption processes at the front end of a communication receiver are sampling, demodulation, filtering and quantization. This energy requirement is modelled as the $\{T_k^r\}$ process. Here we are not including transmission energy by the receiver for further transmission of data. 

 The average energy requirement per slot is $\mathbb{E}[T^r]$.  The receiver harvests $Y_k^r \in \mathcal{Y}^r\subset \mathbb{R}^+$  amount of energy from an ambient energy source at the beginning of  slot $k$. If need be, the entire harvested energy can be used in the same slot. We assume $\{Y_k^r,~ k \geq 1\}$ to be a stationary, ergodic process. The average energy harvested  from the ambient source per slot is $\mathbb{E}[Y^r]$. Since, the receiver can potentially harvest energy from the incoming data stream, the total harvested energy per slot can be greater than $Y_k^r$ and we denote it by $D_k$. In slot $k$, $E_k^r \in  \mathcal{A}^r \subseteq \mathbb{R}^+$ denotes the energy available in the receiver's energy buffer. Let $\hat{E}_k^r \triangleq E_k^r+Y_k^r \in \hat{\mathcal{A}}^r \subseteq \mathbb{R}^+$. Average energy deficit at the receiver $\Delta \triangleq (\mathbb{E}[T^r]-\mathbb{E}[Y^r])^+$ where $(x)^+= \max\{0,x\}$. Thus, the receiver should harness an average of $\Delta$ units of energy per time slot from the incoming RF data stream to meet its operational requirements. Next, we discuss the various receiver architectures that we consider in this work.

\subsubsection{Ideal Receiver}

An ideal receiver can simultaneously harvest energy and \textit{read} the noise corrupted data symbol $W_k$. In every slot $k$, total harvested energy at the receiver is $D_k=\xi_k+Y_k^r$, where $\xi_k=\eta \big(\sum\limits_{k=1}^{L}\sqrt{H_k(i)}X_k(i)\big)^2$ and $\eta$ denotes the energy harvesting efficiency factor (\cite{lu2015wireless}).

\subsubsection{Time-Switching Receiver}
\label{TSR}
A time-switching receiver can either harvest energy or receive the information bearing noise corrupted data symbol, but not both at the same time.  In slot $k$, if $\hat{E}_k^r \geq T_k^r$, then the communication module successfully receives $W_k$. Energy harvested in the slot is $D_k=Y_k^r$. Instead, if $\hat{E}_k^r < T_k^r$, the receiver  merely harvests $\xi_k$ amount of energy from the data stream and stores it in the buffer.  The total energy harvested in the slot is $D_k=Y_k^r+\xi_k$ and the data symbol received is recorded as an erasure. Let $\mathcal{E}_k(t) \triangleq 1\{\hat{E}_k^r < T_k^r\}$, where $1\{.\}$ is the  indicator function. Let ${\mathcal{E}}_k^c(t)=1-\mathcal{E}_k(t)$. We can write the evolution equation for the energy in time switching receiver's energy buffer as $E_{k+1}^r=\hat{E}_k^r+\mathcal{E}_k(t)\xi_k-\mathcal{E}^c_k(t)T_k^r.$ 
  
 \subsubsection{Constant Fraction Power-Splitting Receiver} 
 \label{PSR}
 \par  { We consider power-splitting receivers  which \textit{split} a \textit{fixed fraction} of the received power between the communication receiver and the rectenna in each slot. Power splitter will split a fraction $\pi_{\mathcal{E}}$ of the received power $\xi_k$ such that, under stationarity,  $\pi_{\mathcal{E}}\mathbb{E}[\xi_k]=\Delta$. Then, from user $i$, $\pi_{\mathcal{E}}\mathbb{E}_{\mathbf{H}}[T_i^s(\mathbf{H})]$ is taken by the receiver for energy, and the rest for communication. Equivalently, we can think that transmitter $i$ is communicating through a fading channel with a fading coefficient $\sqrt{\pi_{\mathcal{E}}^c H_k(i)}$ in slot $k$, where $\pi^c_{\mathcal{E}} \triangleq 1-\pi_{\mathcal{E}}$. Making $\pi_{\mathcal{E}}$ adaptively changing with $\mathbf{H}_k$ and $\hat{E}_k^r$ and using Markov Decision Theory can improve the performance of time-switching and power-splitting models. If $\hat{E}_k^r+\pi_{\mathcal{E}}\xi_k < T_k^r$, then the noise corrupted data symbol is registered as an erasure.  Let $\mathcal{E}_k(p)\triangleq 1\{\hat{E}_k^r+\pi_{\mathcal{E}}\xi_k < T_k^r\} $ and $\mathcal{E}_k^c(p)=1-\mathcal{E}_k(p)$. If $\hat{E}_k^r+\pi_{\mathcal{E}}\xi_k \geq  T_k^r$, communication module receives the noise corrupted data symbol $\tilde{W}_k=\sum\limits_{i=1}^{L}\sqrt{\pi^c_{\mathcal{E}}H_k(i)}X_k(i)+N_k$. Total energy harvested in any slot is $D_k=Y_k^r+\pi_{\mathcal{E}}\xi_k$.  We can write the following evolution equation: $E_{k+1}^r=\hat{E}_k^r+\mathcal{E}_k(p)\xi_k-\mathcal{E}^c_k(p)T_k^r.$ }
  
 The energy buffers at the transmitters and the receiver start with some initial energy at time $k=0$. Thus, the process $\{E_k^s(i)\}$ is not necessarily stationary and hence, the process $T^s(i) \equiv \{T_k^s(i)\}$ need not be stationary. We consider processes $T^s(i)$ that belong to the class of Asymptotic Mean Stationary (AMS), ergodic processes    \cite{gray2011entropy}. For a  treatment on the Shannon capacity of a point to point AWGN channel with energy harvesting transmitters adopting AMS, ergodic policies refer \cite{rajesh2014capacity}. Limiting to such a class will not reduce our capacity region.  Since $T^s(i)$ for each $i$ is AMS, ergodic and independent across $i$, $ \{\mathbf{X} _k,~ k \geq 1\}$ is AMS, ergodic where $\mathbf{X}_k \triangleq  (X_k(1), \hdots, X_k(L) )$. In addition, since the channel is memoryless, $\{(\mathbf{X}_k, W_k,T_k^r,E_k^r,Y_k^r ),~ k \geq 1\}$ is jointly AMS, ergodic. 
\section{Fading GMAC with Ideal Receiver} 
\label{S_ID}
{In this section, we characterize the minimum-rate capacity region of a fading GMAC SWIPT system with an ideal receiver}. We consider a block fading model in which the channel remains constant for the duration of a coherence time. The total duration of communication to transmit the codeword spans multiple coherence times. Let $R(i) =\frac{\log(|\mathcal{M}_i|)}{n}$ denote the rate of communication of the $i^{\text{th}}$ transmitter. In addition, let $\rho(i)$ denote the minimum rate at which transmitter $i$ communicates. Note that $\rho(i)$ is a model parameter. The delay sensitive data of transmitter $i$ is to be transmitted at a minimum rate $\rho(i)$. A certain amount of power is required for this. Additional power, if any, is utilized to transmit the delay tolerant data. Note that, if the fading is such that the fading states can take arbitrarily low values, average power required to maintain the minimum rate will be infinity. This is prohibitive. In such scenarios, it make sense to talk of the minimum-rate \textit{outage} capacity region. If the fading process $\{H_k(i),~ k \geq 1\}$ for user $i$ is such that $\mathbb{E}[\frac{1}{H_k(i)}]<\infty$, it is possible to achieve minimum rates in all fading states with a finite average power constraint. In the following, we assume the same. 

We consider the problem of maximizing the rate of transmission of delay tolerant data subject to, maintaining a minimum rate of transmission in each block and meeting an average  power constraint at the transmitters and the receiver. Thus, the problem is one of maximizing certain ergodic rates (\cite{tse1998multiaccess}) subject to maintaining particular delay limited rates (\cite{hanly1998multiaccess}).

In the previous section, we defined energy management policies $T^s(i) $ at each transmitter. The actions corresponding to those policies (energy to be used for transmission), in any slot,  depend on the entire history of energy harvesting process i.e.  $(Y_1^s(i),\hdots Y_k^s(i))$ in addition to the message $M_i$ and joint fading states $\mathbf{H}_k \triangleq \big(H_k(1),\hdots ,H_k(L)\big )$. We now define a sub class of those policies viz. the class of Markov policies. We have the following definition:
\begin{defn}{\emph{(Markov Energy Management Policy):}}
An energy management policy $T_{m,k}^s(i)$ at transmitter $i$, is called a \emph{Markov policy} if $T_{m,k}^s(i): \mathcal{M}_i \times \hat{\mathcal{A}}^s \times \mathcal{H}^L \rightarrow \mathcal{X}$.
\end{defn}
 
In case of Markov policies, at any time, the dependence of action is restricted to the energy in the buffer and energy harvested by the transmitter in that slot, rather than explicitly on  the entire history of energy harvesting process. Note that, because we consider the harvesting processes $\{Y_k^s(i)\}$ to be stationary, ergodic, the policy $ T_{m,k}^s(i)$ may not be Markov (i.e. the processes involved may not be a Markov process) policy in true sense. Nevertheless, we stick on to this terminology to emphasise that the policy is different from the history dependent policies. From now on, we shall denote a Markov policy at transmitter $i$ and slot $k$ as $T^s_k(i)$,  unless explicitly stated otherwise.

{We also make a note of the following standard definitions. If $\mathbf{M}=(M_1,M_2\hdots M_L)$ is the transmitted message vector and $\widehat{\mathbf{M}}$ denotes the decoder's estimate of the same, the average probability of message decoding error (averaged over all random codebooks), denoted as $P_{e}^{(n)}$, is defined to be $Pr\big\{\widehat{ \mathbf{M}} \neq \mathbf{M}  \big\}$. Note that, above definition is made incorporating the probability of the error event that the receiver is not able to receive the entire codeword (corrupted by noise) due to energy deficiency. A rate vector  $\big(R(1),\hdots R(L)  \big)$ is said to be achievable if, for every joint fade vector $\mathbf{h}=\big(h(1),\hdots h(L)\big)$ there exist a sequence of  $\big((2^{nR_1(\mathbf{h})},2^{nR_2(\mathbf{h})},\hdots 2^{nR_L(\mathbf{h})} ),n\big)$ codes, $L$ encoders, a message decoder and an energy receiver such that the probability of message decoding error  $P_{e}^{(n)} \rightarrow 0$ as $n \rightarrow \infty$, the instantaneous rate vector $\mathbf{R}(\mathbf{h}) \triangleq \big(R_1(\mathbf{h}),\hdots R_L(\mathbf{h})\big)$ is such that $R_i(\mathbf{h})\geq \rho(i)$ for all $\mathbf{h}\in \mathcal{H}^L$ and $\mathbb{E}_{\mathbf{H}}[R_i(\mathbf{h})]\geq R(i)$. Minimum-rate capacity region is the closure of the set of all achievable rate vectors.}

\subsection{Minimum-Rate Capacity Region for Ideal Receiver}
We consider energy harvesting transmitters with arbitrarily initialized energy buffer. Since the energy availability at the transmitters is solely dependent on the harvesting process, there can be instances wherein a transmitter is devoid of the required energy to transmit a codeword symbol. Thus the energy harvesting constraints further  restrict the minimum rates that are admissible. In particular, note that, without the energy harvesting constraint, the delay sensitive data can be sent at any rate within the delay limited capacity of the corresponding channel with average power constraints. The energy harvesting process puts an instantaneous, time varying peak power constraint on the transmitted symbols. Hence, the minimum rate $\rho(i)$ for each transmitter $i$ has to be chosen such that the rate vector $\bm{\rho}~\triangleq \big(\rho(1),\hdots,\rho(L)\big)$ is within the corresponding delay limited capacity region of all peak power constrained channels induced by the energy harvesting process. We denote,
\begin{flalign*}
\hspace{-20pt}
\mathcal{C}(\mathbf{T}^s)  \triangleq  \Big\{\mathbf{R}:~ & \bm{\rho}(A) \leq \mathbf{R}(A) \leq  \mathbb{E}_\mathbf{H}\Big[\mathbf{C}_{A}(\mathbf{H})\Big],~ \forall A \Big\},
\end{flalign*}
where $\mathbf{C}_{A}(\mathbf{H}) \triangleq \frac{1}{2}\log\Big(1+\frac{1}{\sigma^2}\sum\limits_{i \in A}H(i)T_i^s(\mathbf{H})\Big)$,~ $\bm{\rho}(A) \triangleq \sum\limits_{i \in A}\rho(i)$, $\mathbf{R}(A) \triangleq \sum\limits_{i \in A}R_i$ and $T_i^s(\mathbf{H})$ (or simply $T_i^s$) denotes a Markov energy management policy adopted at transmitter $i$ when the joint fading vector is $\mathbf{H}$ and $ A \subset [1:L]$. Also, let $$\mathcal{T}^s_m(\Delta)  \triangleq  \Big\{\mathbf{T}^s: \mathbb{E}_\mathbf{H}\big[T_i^s(\mathbf{H})\big]\leq \mathbb{E}[Y^s(i)],~R_i(\mathbf{H}) \geq \rho(i),$$ $$ \sum\limits_{i=1}^{L}\mathbb{E}_\mathbf{H}\big[H(i)T_i^s(\mathbf{H})\big] \geq \frac{\Delta}{\eta},~ i \in [1:L]\Big\},$$ where $\mathbf{T}^s=(T_1^s,\hdots T_L^s)$ ($m$ stands for Markov). 

Now, we provide a characterization of the minimum rate capacity region.

\begin{thm}
\label{cap_region}
\emph{(Minimum-Rate Capacity Region for Ideal Receiver)}: The minimum rate capacity-energy region is 
\begin{flalign*}
\hspace{35pt}\mathcal{C}(\Delta) = \overline{\text{Conv}}\Bigg(\bigcup_{\mathbf{T}^s \in \mathcal{T}^s_m(\Delta)} \mathcal{C}(\mathbf{T}^s)\Bigg),
\end{flalign*}
where $\overline{\text{Conv}}(A)$ denotes the closure of convex hull of $A$.
\end{thm}
\begin{proof}
See Appendix A.
\end{proof}
\begin{rem}
\label{rem1}
 For the special case of GMAC with constant channel gain, say $(g_1,\hdots,g_L)$, the constraint $\sum\limits_{i=1}^{L} \mathbb{E}_\mathbf{H}\big[H(i)T_i^s(\mathbf{H})\big] \geq \Delta / \eta$ is equivalent to $\Delta/ \eta \leq  \sum\limits_{i=1}^{L}g_i\mathbb{E}[Y^s(i)]$. If the model parameters are such that this condition is met, the constraint is innocuous. If $\Delta/ \eta >  \sum\limits_{i=1}^{L}g_i\mathbb{E}[Y^s(i)]$, energy deficit at the receiver can be mitigated (at the expense of reduction in data rate) by additional transmitter cooperation \cite{fouladgar2012transfer}. 
\end{rem}

\begin{rem}
Note that the motivation to consider minimum rate constrained communication is to ensure the transmission of a delay sensitive data at a certain minimum rate, irrespective of the fading states. One way to implement the simultaneous transmission of delay sensitive and delay tolerant data, is to employ two separate codebooks-one for the delay sensitive data and the other for the delay tolerant data. The decoder at the end of every block, decodes the delay sensitive data by, treating the codeword corresponding to delay tolerant data and the codeword corresponding to the delay sensitive data of the other user, as interference. The minimum rates permissible might be low due to interference.
\end{rem}
\subsection{Explicit Characterization of Minimum-Rate Capacity Region}
Next, we turn our attention to the problem of explicitly characterizing the boundary points of the minimum-rate capacity-energy region of the channel model under consideration. Such a characterization was first obtained for a fading GMAC in \cite{tse1998multiaccess} (\textit{Lemma} 3.10) (without energy harvesting). A similar revenue optimization problem with minimum rate guarantees was pursued in \cite{goyal2002optimal}. In this section, we derive the corresponding results  for the minimum rate capacity region of a fading GMAC with energy harvesting transmitters and receiver.  

 We define the vector $\bm{\mu}\triangleq \big(\mu(1),\hdots\mu(L)\big)$ such that $\mu(i)  \in \mathbb{R}^+,~ \forall i$. As in \cite{tse1998multiaccess}, we prioritize the users based on the vector $\bm{\mu}$ (also called the rate reward vector). The characterization of the  minimum rate capacity-energy region is in terms of this parameter $\bm{\mu}$ (i.e., different points of the boundary can be obtained by choosing the $\bm{\mu}$ vector appropriately and solving the following optimization problem).  We refer $\mathbf{R}(\mathbf{H}) \triangleq \big(R_1(\mathbf{H}),\hdots R_L(\mathbf{H})\big)$ to as the instantaneous rate vector corresponding to the joint fade state $\mathbf{H}$. The boundary points of the minimum-rate capacity region corresponds to those rates vectors such that no user can unilaterally increase its rate while the rates of other users remain fixed, and remain within the minimum rate capacity energy region. Owing to the convexity of the the region $\mathcal{C}(\Delta)$, the boundary points can be obtained by solving the following optimization problem. We consider the following optimization problem : 
 \begin{flalign*}
\hspace{20pt}
&\max_{\mathbf{T}^s(.)} ~ \sum\limits_{i=1}^L \mu(i)\Big( \mathbb{E}_{\mathbf{H}}\big[R_i(\mathbf{H})\big] \Big),\\
& \text{subject to}~ \mathbb{E}_\mathbf{H}[T^s_i(\mathbf{H})] \leq \mathbb{E}[Y^s(i)], ~ \forall i, \\
& R_i(\mathbf{h}) \geq \rho(i),~\forall ~\mathbf{h},~ i,\\
&\sum\limits_{i=1}^{L} \mathbb{E}_\mathbf{H}\big[H(i)T_i^s(\mathbf{H})\big] \geq \frac{\Delta}{\eta}.
\end{flalign*}
We can make use of the following lemma (which is a variant of \textit{Lemma} 2, \cite{goyal2002optimal}) so as to compute the boundary points of the minimum-rate capacity region. 

\begin{lem}\emph{(The Lagrangian Condition):}
\label{lag_goyal}
Let $\mathbf{T}^*(.)$ be the solution to the optimization problem stated. For a given $\bm{\mu}$, $\mathbf{T}^*(.)$ is a solution to the optimization problem if and only if there exists {a $\bm{\lambda} \triangleq \big(\lambda_s(1),\hdots,\lambda_s(L),\lambda_r\big),~ \lambda_s(i),~\lambda_r \in \mathbb{R}^+$ and an energy management policy $\mathbf{T}^s(.)$ such that for every joint fade state $\mathbf{h}$, $\mathbf{T}(\mathbf{h})$ is a solution to the optimization problem 
\begin{flalign*}
&\max_{\mathbf{t}^s} \sum_{i=1}^L \mu(i) r(i) -\lambda_s(i) t^s(i)+\lambda_r h(i)t^s(i), \\
& \emph{s.t} ~ \mathbf{r} \in \mathcal{C}\big(\mathbf{h}, \mathbf{t}^s \big),~\mathbb{E}_{\mathbf{H}}\big[R_i(\mathbf{H})\big] =R^*(i),\\
&~\mathbb{E}_{\mathbf{H}}\big[H(i)T_i^s(\mathbf{H})\big]\geq \Delta/\eta,  ~\mathbb{E}_{\mathbf{H}}\big[T_i^s(\mathbf{H})\big]=\mathbb{E}\big[Y^s(i)\big],
\end{flalign*}
where, }
\begin{flalign*}
\hspace{-20pt}
\mathcal{C}\big(\mathbf{h}, \mathbf{t}^s \big)  \triangleq  \Big\{\mathbf{r}:~ & \bm{\rho}(A) \leq \mathbf{r}(A) \leq c_A(\mathbf{h},\mathbf{t}),~ A \subset [1:L]\Big\},
\end{flalign*}
$c_A(\mathbf{h},\mathbf{t})\triangleq \frac{1}{2}\log\Big(1+\frac{1}{\sigma^2}\sum\limits_{i \in A}h(i)t^s(i)\Big)$, $R^*(i)$ corresponds to the average rate of user $i$ adopting the optimal energy management policy $T^*_i(.)$ and $\mathbf{r}(A)=\sum\limits_{i \in A}r(i)$.\qed 
\end{lem}

\section{Fading GMAC with Time-Switching Receiver}
\label{S_TS}
{In this section, we characterize the minimum-rate \textit{erasure} capacity region of a fading GMAC with energy harvesting constraints, for the time-switching receiver case. In our setting, transmitters are oblivious to the realizations of the $\{T_k^r\}$ process at the receiver. Thus, we do not assume any coordination between the transmitters and the receiver in deciding when to harvest RF energy. We derive the rate regions with this assumption in place. In particular, we will show that the receiver can independently and identically decide to harvest energy or not from the incoming data stream. We will prove, this alone suffices to mitigate the energy outages in the system, asymptotically. As harvesting energy results in the erasure of data symbols, we model the resultant channel as a fading GMAC with energy harvesting transmitters and receiver in conjunction with an erasure channel. The decrease in channel capacity due to erasures at the output of a noisy channel was previously characterized in \cite{verdu2008information}. Let $\pi_{\mathcal{E}}=\Delta/\sum\limits_{i=1}^{L} \mathbb{E}_\mathbf{H}\big[\eta H(i)T_i^s(\mathbf{H})\big]$  be the probability with which receiver harvests from RF stream independently across slots and let $\pi^c_{\mathcal{E}}=1-\pi_{\mathcal{E}}$.} 
\subsection{  Erasure Capacity Region with Time Switching Receiver}
\label{tsr_section}
 Let 
\begin{flalign*}
\hspace{-20pt}
\mathcal{C}^e_t(\mathbf{T}^s)  \triangleq & \Big\{\mathbf{R}:~  \bm{\rho}(A) \leq \mathbf{R}(A) \leq  \mathbb{E}_\mathbf{H}\Big[\mathbf{C}_{t,A}(\mathbf{H})\Big],~\forall A\Big\},
\end{flalign*}
where $\mathbf{C}_{t,A}(\mathbf{H})\triangleq \frac{\pi_{\mathcal{E}}^c}{2}\log\Big(1+\frac{1}{\sigma^2}\sum\limits_{i \in A}H(i)T_i^s(\mathbf{H})\Big)$,  $A \subset [1:L]$. We have the following result.

\begin{thm}
\label{Th_TSR}
\emph{(Minimum-Rate Erasure Capacity Region)}: 
\begin{flalign*}
\hspace{35pt}\mathcal{C}_t^e(\Delta) = \overline{\text{Conv}}\Bigg(\bigcup_{\mathbf{T}^s \in \mathcal{T}^s_m(\Delta)} \mathcal{C}_t^e(\mathbf{T}^s)\Bigg),
\end{flalign*}
is the minimum rate erasure-capacity region.
\end{thm}

\begin{proof}
See Appendix B.
\end{proof}

\begin{rem}
The total mean harvested energy at the time switching receiver, under stationarity, is $\mathbb{E}[Y^r]+\pi_{\mathcal{E}}\mathbb{E}[\xi]$. The choice of $\pi_{\mathcal{E}}$ is made so as to ensure the receiver, on an average, harvests more energy than it uses up. The potential rate loss due to the harvesting constraint at receiver is captured in the erasure probability term $\pi_{\mathcal{E}}$.
\end{rem}

\begin{rem}
We can obtain an explicit characterization of the above minimum-rate erasure capacity region by obtaining a Lagrangian condition as in Lemma \ref{lag_goyal}. As the proof steps are identical, we do not restate it.
\end{rem}
\section{Fading GMAC with Constant Fraction Power-Splitting Receiver}
\label{S_PS}
In this section, we extend our previous discussions to the case of a fading GMAC with energy harvesting constraints and power-splitting receiver introduced in Section \ref{PSR}. Define $\pi_{\mathcal{E}}$ as in the previous section. 
\subsection{ Capacity-Energy Region with Power Splitting Receiver: Constant Fraction Splitting}
 We define 
\begin{flalign*}
\hspace{-20pt}
\mathcal{C}^e_p(\mathbf{T}^s)  \triangleq  \Big\{\mathbf{R}:~ & \bm{\rho}(A) \leq \mathbf{R}(A) \leq  \mathbb{E}_\mathbf{H}\Big[\mathbf{C}_{p,A}(\mathbf{H})\Big],~\forall A\Big\},
\end{flalign*}
$\mathbf{C}_{p,A}(\mathbf{H}) \triangleq \frac{1}{2}\log\Big(1+\frac{\pi_{\mathcal{E}}^c}{\sigma^2}\sum\limits_{i \in A}H(i)T_i^s(\mathbf{H})\Big)$, $ A \subset [1:L]$. We have the following result.

\begin{thm}
\label{Th_PSR1}
\emph{(Minimum-Rate Capacity Region via Constant Fraction Splitting)}: The closure of
\begin{flalign*}
\hspace{35pt}\mathcal{C}_{p}^e(\Delta) = \overline{\text{Conv}}\Bigg(\bigcup_{\mathbf{T}^s \in \mathcal{T}^s_m(\Delta)} \mathcal{C}_p^e(\mathbf{T}^s)\Bigg),
\end{flalign*}
 is the minimum-rate capacity region with constant fraction power splitting receiver.
\end{thm}

\begin{proof}
Except for a re-scaling of the signal to noise ratio at the receiver, the proof follows along the lines of the proof of Theorem \ref{cap_region}.
\end{proof}

\section{Numerical Results}
\label{S_NR}
In this section, we compare the capacity region with minimum rate constraints of a fading GMAC with energy harvesting transmitters and receiver, for the ideal, time-switching and power-splitting receive models. Consider a two user GMAC with average harvested energy of user 1 per slot, $\mathbf{E}[Y^s(1)]=5$W and of user 2, $\mathbf{E}[Y^s(2)]=3$W. The fading distribution at each user is a discrete distribution obtained as follows: Consider Rayleigh distribution with parameter $\alpha=1$. Fix a quantization size $q=.1$ and a maximum fading coefficient $h_{\text{max}}=5$. The support set of the distribution is taken to be $\big[q,2q,\hdots, [h_{\text{max}}]\big]$, where $[h_{\text{max}}]$ is $h_{\text{max}}$ approximated to the greatest integer multiple of $q$ less than or equal to $h_{\text{max}}$. The probability $p_{h}(i)$ of the $i^{\text{th}}$ element of the support set is fixed to be the probability of a Rayleigh random variable taking values in the interval $[(i-1)q,~iq]$. We consider i.i.d fading and fading is independent across users. Let $\rho(1)=0.3$ bits/channel use,  $\rho(2)=0.2$ bits/channel use. Let the average energy harvested from the ambient source at the receiver $\mathbb{E}[Y^r]=10\mu$W and the average energy requirement $\mathbb{E}[T^r]=20 \mu$W. Hence, the energy deficit at the receiver $\Delta=10\mu$W=-20dBm. The efficiency factor $\eta$ is fixed to $10^{-5}$. Each time slot is fixed to be of $1\mu$ second duration.

{With system parameters fixed as above, we compute the minimum-rate capacity region for various receiver models under consideration (Figure \ref{rate_reg}). We compare these regions with  the capacity region without minimum rate constraints for the ideal receiver and the rate achievable \textit{without} RF power transfer. It is observed that RF power transfer enhances the rate region considerably. Also, for the same system parameters, the SWIPT system with power splitting receiver achieves better \textit{throughput} compared to that with time switching receiver. }
\begin{figure}[h]
\begin{center}
\includegraphics[scale=0.45]{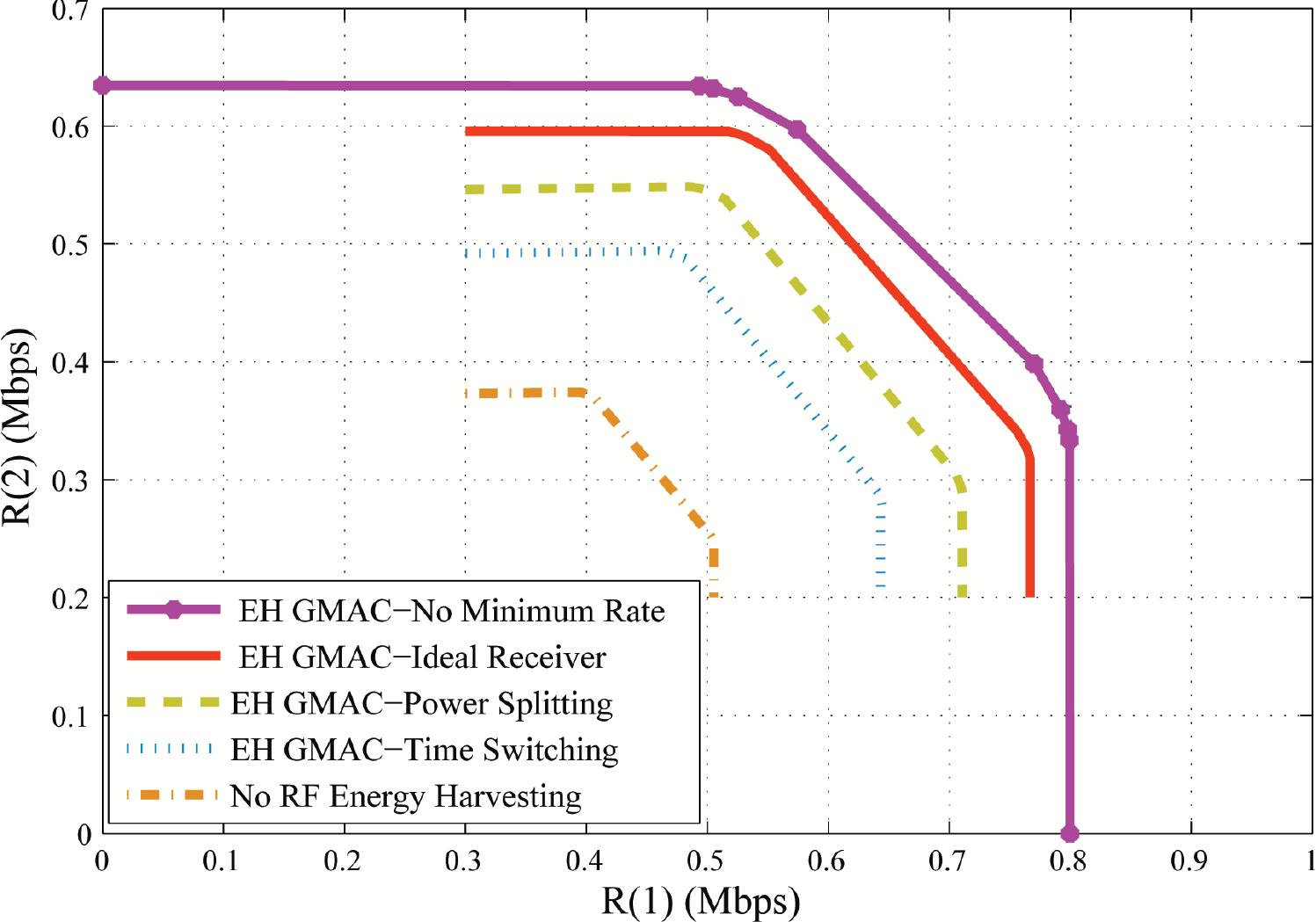}
\caption{Comparison of capacity regions for various receiver architectures } \label{rate_reg}
\end{center}
\end{figure}

Next, in Figure \ref{sum_rate}, we plot the maximum sum rates achievable with each of the receiver models. We fix the minimum rates $\rho(1)=\rho(2)=.1$ bits/channel use and $\mathbb{E}[Y^s(1)]=\mathbb{E}[Y^s(2)]$. In Figure \ref{sum_rate_delta}, we provide a comparison of the sum rates versus the energy deficit at the receiver. Enhancement of the sum rate due to RF power transfer, in particular, using power splitting receiver is observed. The dominance of power-splitting receivers can be attributed to the concavity of the optimal rates as a function of the power used for coding. We also plot the shrinkage in rate-region for the power-splitting and time-switching receivers as the energy deficit at the receiver increases (Figure \ref{rr_PS_D} and Figure \ref{rr_TS_D}, respectively). 
\begin{figure}[h]
\begin{center}
\includegraphics[scale=0.45]{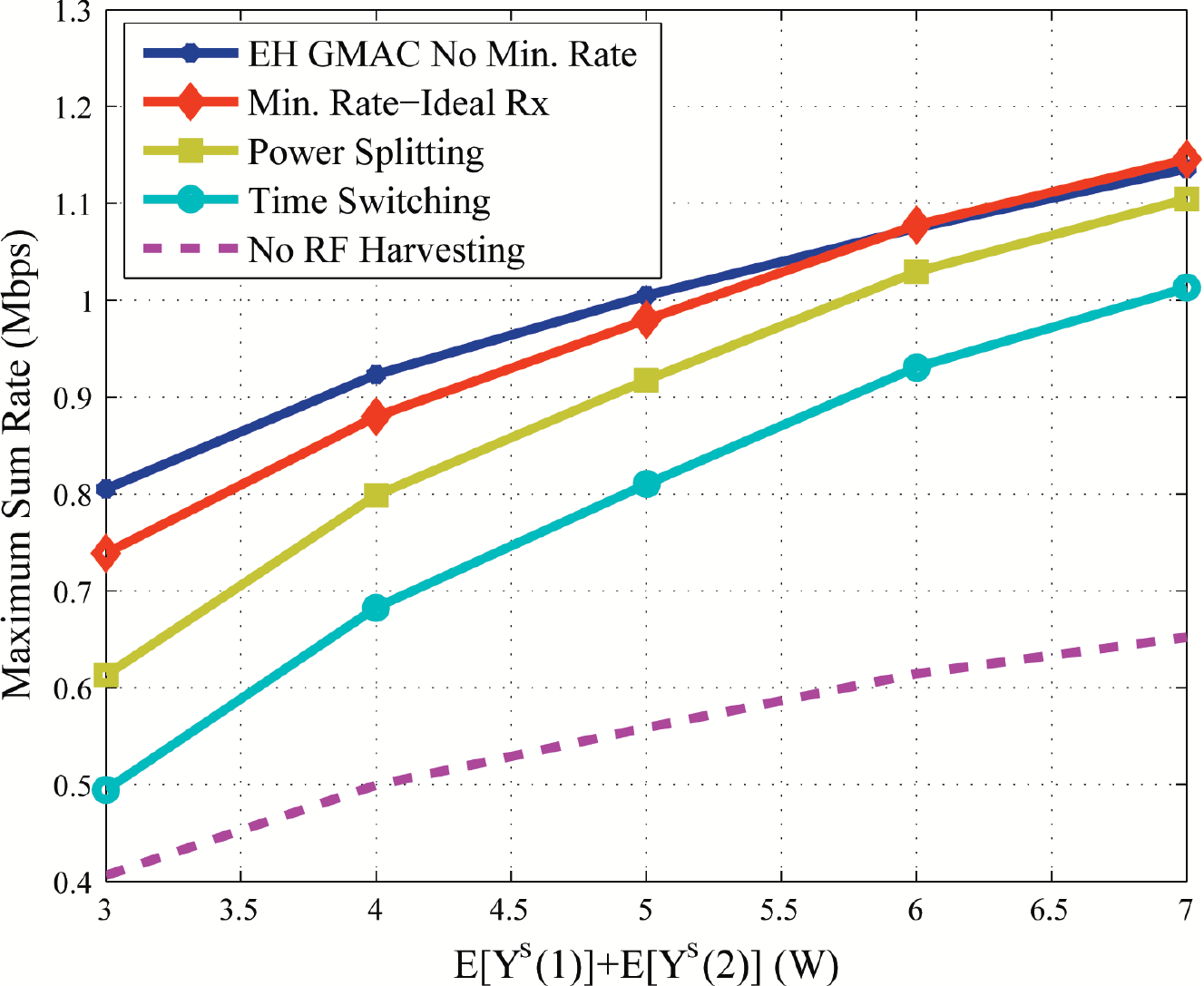}
\caption{Comparison of sum rates for various receiver architectures} \label{sum_rate}
\end{center}
\end{figure}   

\begin{figure}[h]
\begin{center}
\includegraphics[scale=0.45]{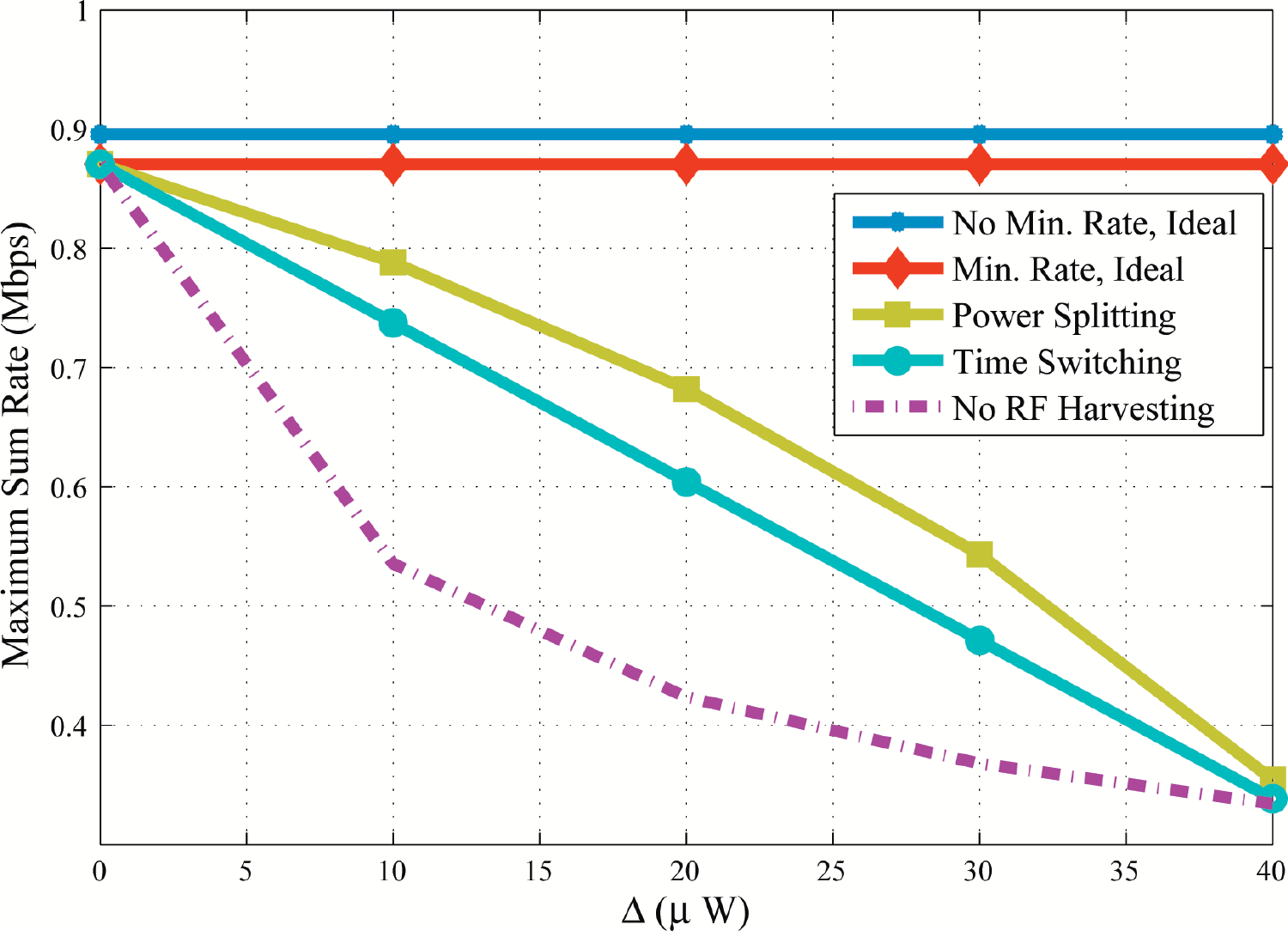}
\caption{Comparison of sum rates versus energy deficit at the receiver for various receiver architectures} \label{sum_rate_delta}
\end{center}
\end{figure}

\begin{figure}[h]
\begin{center}
\includegraphics[scale=0.45]{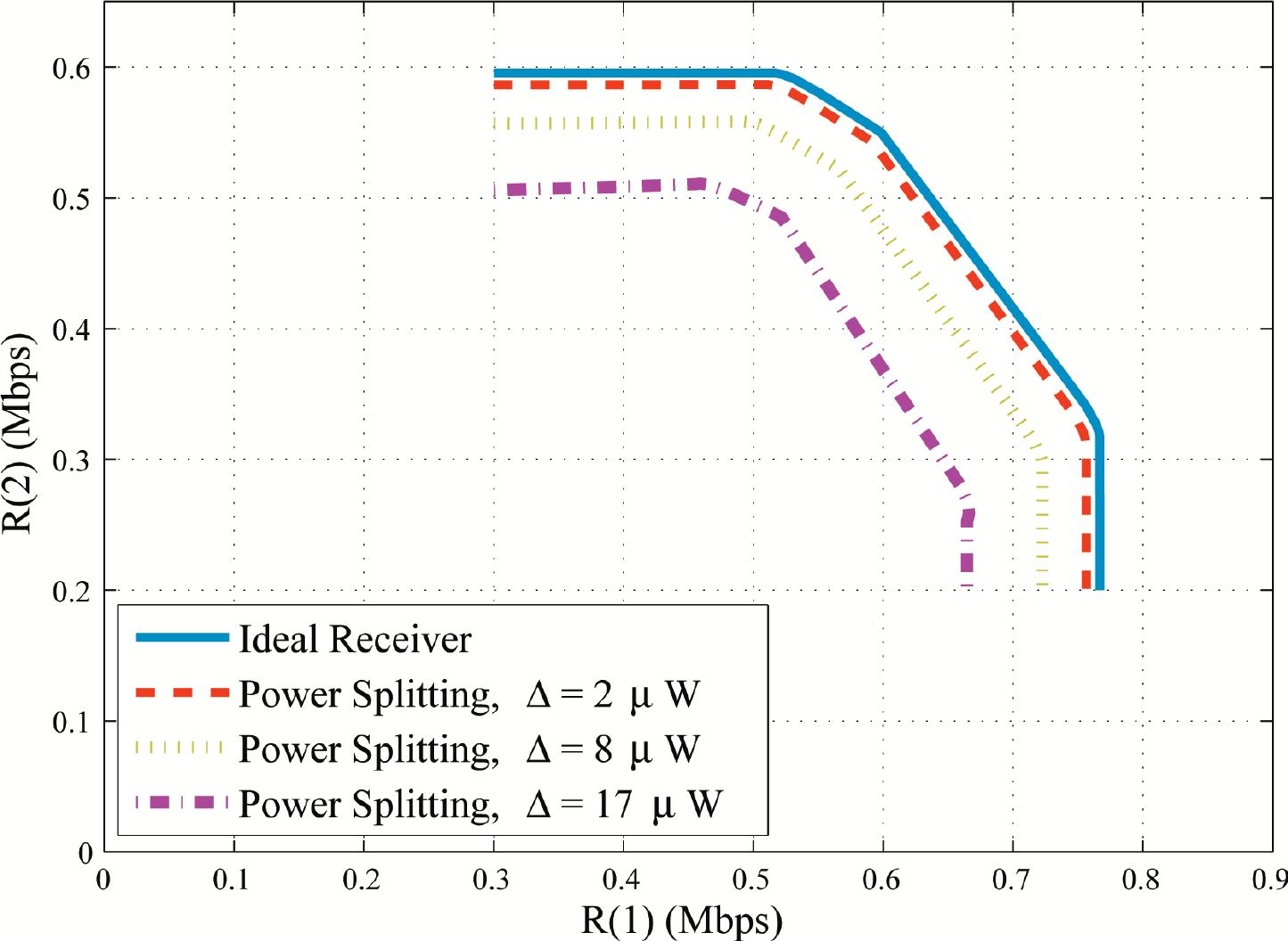}
\caption{Comparison of rate regions for the power-splitting receiver model for various values of energy deficit at the receiver} \label{rr_PS_D}
\end{center}
\end{figure}   

\begin{figure}[h]
\begin{center}
\includegraphics[scale=0.45]{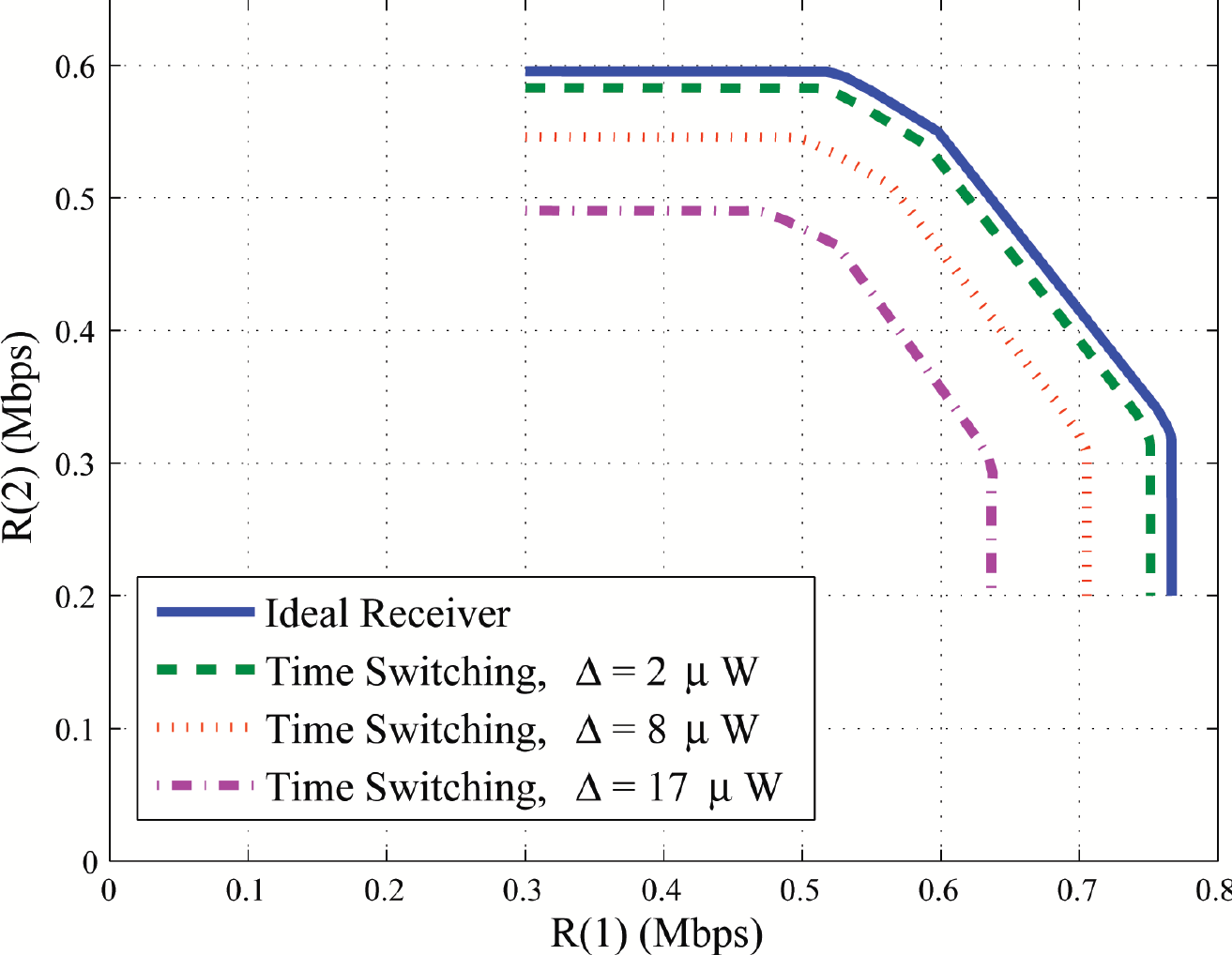}
\caption{Comparison of rate regions for the time-switching receiver model for various values of energy deficit at the receiver} \label{rr_TS_D}
\end{center}
\end{figure}

\section{Conclusion}
\label{S_Conc}
In this work, first we characterized the fundamental limits of communication when delay limited and delay sensitive data are simultaneously transmitted over a fading GMAC SWIPT system. We characterized the minimum-rate capacity region a) when the receiver is assumed to be ideal, b) when the receiver is time switching or power splitting. The results were obtained without any additional assumption on the synchronization between the transmitters and the receiver so as to harvest RF power at the receiver.  
\section*{Appendix A}
\textit{Proof of Theorem \ref{cap_region}:}\\
\textit{Achievability:}\\
\textit{Encoding:}
Fix a rate vector $\mathbf{R}$. In addition, we are given a minimum rate vector $\bm{\rho}$. As noted in \cite{tse1998multiaccess}, instead of generating codebooks corresponding to each joint fading state, we can generate fixed codebooks independent of the fading states. In any slot $k$, consider $\sqrt{T_k^s(i)H_k(i)}$ to be the fading process. At each transmitter $i$, codebooks $\mathcal{C}(i)$ are generated i.i.d according to the distribution $\mathcal{N}(0,1)$. Thus $\mathcal{C}(i)$ denotes a random Gaussian matrix of dimension $2^{nR_i}\times n$. These codebooks are shared among all transmitters and with the receiver.   We assume that all the transmitters and the receiver can track the realizations of the fading process $\big\{\mathbf{H}_k\big\}$. And hence, the receiver can make decoding decisions given the codebooks and fading realizations.  

Next, we define the following  specific Markov energy management policy (we call it the truncated Markov policy):
 \begin{displaymath}
   T_i^s(\mathbf{H}_k) = \left\{
     \begin{array}{lr}
       T'_i(\mathbf{H}_k) & :  T'_i(\mathbf{H}_k) \leq \hat{E}_k^s(i)\\
       \hat{E}_k^s(i) & : T'_i(\mathbf{H}_k) > \hat{E}_k^s(i),
     \end{array}
   \right.
\end{displaymath}
where $T'_i(\mathbf{H}_k)$ is the optimal power allocation policy  at transmitter $i$ for the fading multiple access channel without energy harvesting constraints and having minimum rate constraints. Characterization of $T'_i(H_k)$ can be obtained solving the optimization problem in Lemma \ref{lag_goyal}. The Lagrangian multipliers therein are obtained by  solving $\mathbb{E}\big[T'_i(\mathbf{H})\big]=\mathbb{E}[Y(i)]-\epsilon$, for some small $\epsilon>0$ for each $i$. The channel input of transmitter $i$ at time $k$  is
 \begin{flalign*}
 \hspace{-20pt}
 &X_{k}(i)= \text{sgn}\big(X_{k}'(i)\big) \min\Big(|X_{k}'(i)|\sqrt{T_{i}^s(\mathbf{H}_k)},\sqrt{\hat{E}_k^s(i)}\Big),
 \end{flalign*}
 where $\text{sgn}(x)$ is the following function:
 \begin{displaymath}
  \text{sgn}(x) = \left\{
     \begin{array}{lr}
      +1, & \text{for}~  x \geq 0,\\
      -1, & \text{for}~ x < 0.
     \end{array}
   \right.
\end{displaymath}

The codebooks generated have a average power constraint of $\mathbb{E}[Y(i)]-\epsilon$ at each transmitter. Whereas, the average energy replenished in each slot at transmitter $i$ is $\mathbb{E}[Y(i)]$. It follows that as $k \rightarrow \infty$, $E_k^s(i) \rightarrow \infty$ a.s. (refer Chapter 7, \cite{walrand1988introduction}). Hence, $|X_k(i) -X'(i) \sqrt{T'(i)}| \rightarrow 0$ a.s. as $k \rightarrow \infty$, where $X'(i) \sim \mathcal{N}(0,\sigma^2)$, $T'(i)=T_i'(.)$ as defined previously. Thus, $\big\{{X}_k(i)\big\}$ sequence is AMS, ergodic with stationary mean, the distribution of the process $\{\sqrt{T'_k(i)} X_k'(i),~k \geq 1\} $. 

\subsubsection*{Decoding}
The receiver is entrusted with harvesting energy and receiving data from the incoming RF signal. Since the receiver is assumed to be ideal, energy is harvested without corrupting the data symbol. The decoder adopts joint typicality decoding. Upon receiving $W^n$, the decoder scales  each of the codebook according to the corresponding energy management policy values (this can be done as the fading values are known to the decoder). The decoder finds a unique codeword vector $\big(X'^{n}(\hat{M}_1),\hdots, X'^{n}(\hat{M}_L)\big)$ such that $\Big(\big(X'^{n}(\hat{M}_1),\hdots, X'^{n}(\hat{M}_L)\big),W^{n}\Big) \in A_\epsilon^{(n)}$, where $ A_\epsilon^{(n)}$ is the set of all joint weakly $\epsilon-$typical sequences with respect to the stationary mean of the underlying joint AMS process.
If there exists, a unique such vector $\big(\hat{M}_1,\hdots \hat{M}_L\big)$, the decoder decides it to be the transmitted message vector. Else, an error is declared. 
\subsubsection*{Analysis of Error Events}
Owing to the symmetry of codebook construction, we can assume that $(M_1=1,\hdots,M_L=1)$ is s sent. Following conditional error events can happen:

$\mathbf{E1:}$ $\Big\{ \Big(\big(X'^{n}(1),\hdots X'^{n}(1)\big),W^{n}\Big) \notin A_\epsilon^{(n)}  \Big\}$. Since, AEP holds for AMS, ergodic sequences (\cite{barron1985strong}), we can show that the probability of the error event goes to zero asymptotically. In particular, finiteness of the conditional mutual information provided in the hypothesis of Theorem 3, \cite{barron1985strong} can be verified to hold good using the properties of the general entropy function provided Chapter 7, \cite{gray2011entropy}), non-negativity of mutual information and the fact that the random variables involved have finite second moment. Also, the stationary mean of the AMS process is absolutely continuous with respect an i.i.d Gaussian measure on a suitable Eucledian space. Thus conditions necessary for the AEP result to hold are satisfied.

$\mathbf{EA}:$ $\bigcup\limits _{(m_1,\hdots m_L)}\Big\{ \Big(\big(X'^{n}(\hat{m}_1),\hdots ,X'^{n}(\hat{m}_L)\big),W^{n}\Big) \in A_\epsilon^{(n)}: \hat{m}_k=1,~\text{for}~k\in A~ \text{and}~ \hat{m}_k\neq 1,~\text{for}~k\in A^c \Big\}$, for each $ A \subset [1:L]$. The decoding is done with respect to the stationary mean distribution of the corresponding underlying AMS process. Each finite dimensional distribution corresponding to the stationary mean is an i.i.d  distribution.  Noting this point, the analysis follows as in the standard case. It can be shown that probability of each of these conditional error events, indexed by each subset $A$ of $[1:L]$ can be driven to zero if $R(A) < \mathbb{E}_\mathbf{H}\Big[\frac{1}{2}\log\Big(1+\frac{1}{\sigma^2}\sum\limits_{i \in A}H(i)T_i^s(\mathbf{H})\Big)\Big]$. This proves the achievability result.

\subsection*{Converse:} To prove the converse part, assume that there exist codebooks, encoders and decoders such that $P^{(n)}_e$ (average probability of decoding error)  goes to zero as $n \rightarrow \infty$. We have that $\frac{1}{n} \sum\limits_{k=1}^{n}T_k^s(i) \leq \frac{1}{n} \sum\limits_{k=1}^{n}Y_k(i)\leq \mathbb{E}[Y(i)]+\epsilon $ for large $n$ and arbitrarily chosen small $\epsilon>0$. Hence, along the lines of the converse proof for fading GMAC without energy harvesting constraints (\cite{tse1998multiaccess}), we get $R(A) \leq \mathbb{E}\Big[\frac{1}{2}\log\Big(1+\frac{1}{\sigma^2}\sum\limits_{i \in A}H(i)T_i^s\big( \mathbf{H}\big)\Big)\Big] $ for any subset $A \subset [1:L]$. 

Finally, combining the direct and converse part, we get the required result. \qed

\section*{Appendix B}
\textit{Proof of Theorem \ref{Th_TSR}:}
Fix $\pi_{\mathcal{E}}$ as mentioned in the beginning of Section \ref{tsr_section}. The receiver switches between harvesting energy and receiving noise corrupted data symbol from the channel output randomly according to $\pi_{\mathcal{E}}$. In slot $k$, if the Bernoulli random variable $B_k=1$ (generated independent of all other random variables involved), the receiver harvests energy from the channel output symbol. Probability of this happening is $\pi_{\mathcal{E}}$. The channel output is then recorded as an erasure. Thus, a fading GMAC with time switching receiver, can be equivalently thought of as a fading GMAC with ideal receiver observed through an erasure channel. The erasure instances correspond to the energy harvesting instances at the receiver. The proof of achievability and converse of the minimum rate capacity-energy of a fading GMAC with the ideal receiver is provided in Appendix A. In addition, the capacity of a noisy channel observed through erasure channel is provided in \cite{verdu2008information}. 

The random coding achievability proof of Appendix A with the rate of codebook of transmitter $i$ chosen to be $ \mathbb{E}\Big[\frac{\pi_{\mathcal{E}}^c}{2}\log\Big(1+\frac{1}{\sigma^2}H(i)T_i^s\big( \mathbf{H}\big)\Big)\Big]-\epsilon$, for some small $\epsilon>0$ can be repeated to  show the achievability part. The encoders do not know the locations of erasures a priori. The decoder simply discards the erased channel output symbols. As the blocklength tends to $\infty$, the probability of decoding error events can be shown to go to zero, as in Appendix A.

To prove the converse, we can assume that the encoder has access to non-causal knowledge of erasure locations. The encoders can choose to send a zero symbol during the erasure instances. The decoder discards the erased channel outputs. The converse can be proved by extending in a standard way to MAC setting, the argument given in \cite{verdu2008information}. \qed

\bibliographystyle{IEEEtran}
\bibliography{bibfile_fb_capacity}

\begin{thebibliography}{10}
\providecommand{\url}[1]{#1}
\csname url@samestyle\endcsname
\providecommand{\newblock}{\relax}
\providecommand{\bibinfo}[2]{#2}
\providecommand{\BIBentrySTDinterwordspacing}{\spaceskip=0pt\relax}
\providecommand{\BIBentryALTinterwordstretchfactor}{4}
\providecommand{\BIBentryALTinterwordspacing}{\spaceskip=\fontdimen2\font plus
\BIBentryALTinterwordstretchfactor\fontdimen3\font minus
  \fontdimen4\font\relax}
\providecommand{\BIBforeignlanguage}[2]{{%
\expandafter\ifx\csname l@#1\endcsname\relax
\typeout{** WARNING: IEEEtran.bst: No hyphenation pattern has been}%
\typeout{** loaded for the language `#1'. Using the pattern for}%
\typeout{** the default language instead.}%
\else
\language=\csname l@#1\endcsname
\fi
#2}}
\providecommand{\BIBdecl}{\relax}
\BIBdecl

\bibitem{varshney2008transporting}
L.~R. Varshney, ``Transporting {I}nformation and {E}nergy {S}imultaneously,''
  in \emph{Information Theory, 2008. ISIT 2008. IEEE International Symposium
  on}.\hskip 1em plus 0.5em minus 0.4em\relax IEEE, 2008, pp. 1612--1616.

\bibitem{lu2015wireless}
X.~Lu, P.~Wang, D.~Niyato, D.~I. Kim, and Z.~Han, ``Wireless {N}etworks with rf
  {E}nergy {H}arvesting: A {C}ontemporary {S}urvey,'' \emph{IEEE Communications
  Surveys \& Tutorials}, vol.~17, no.~2, pp. 757--789.

\bibitem{grover2010shannon}
P.~Grover and A.~Sahai, ``Shannon {M}eets {T}esla: {W}ireless {I}nformation and
  {P}ower {T}ransfer.'' in \emph{ISIT}, 2010, pp. 2363--2367.

\bibitem{zhang2013mimo}
R.~Zhang and C.~K. Ho, ``Mimo {B}roadcasting for {S}imultaneous {W}ireless
  {I}nformation and {P}ower {T}ransfer,'' \emph{Wireless Communications, IEEE
  Transactions on}, vol.~12, no.~5, pp. 1989--2001, 2013.

\bibitem{zhou2013wireless}
X.~Zhou, R.~Zhang, and C.~K. Ho, ``Wireless {I}nformation and {P}ower
  {T}ransfer: Architecture {D}esign and {R}ate-{E}nergy {T}radeoff,''
  \emph{Communications, IEEE Transactions on}, vol.~61, no.~11, pp. 4754--4767,
  2013.

\bibitem{fouladgar2012transfer}
A.~M. Fouladgar and O.~Simeone, ``On the {T}ransfer of {I}nformation and
  {E}nergy in {M}ulti-user {S}ystems,'' \emph{Communications Letters, IEEE},
  vol.~16, no.~11, pp. 1733--1736, 2012.

\bibitem{amor2015feedback}
S.~B. Amor, S.~M. Perlaza, I.~Krikidis, and H.~V. Poor, ``Feedback {E}nhances
  {S}imultaneous {W}ireless {I}nformation and {E}nergy {T}ransmission in
  {M}ultiple {A}ccess {C}hannels,'' \emph{arXiv preprint arXiv:1512.01465},
  2015.

\bibitem{hadzi2014multiple}
Z.~Hadzi-Velkov, N.~Zlatanov, and R.~Schober, ``Multiple-{Ac}cess {F}ading
  {C}hannel with {W}ireless {P}ower {T}ransfer and {E}nergy {H}arvesting,''
  \emph{Communications Letters, IEEE}, vol.~18, no.~10, pp. 1863--1866, 2014.

\bibitem{krikidis2014simultaneous}
I.~Krikidis, S.~Timotheou, S.~Nikolaou, G.~Zheng, D.~W.~K. Ng, and R.~Schober,
  ``Simultaneous {W}ireless {I}nformation and {P}ower {T}ransfer in {M}odern
  {C}ommunication {S}ystems,'' \emph{Communications Magazine, IEEE}, vol.~52,
  no.~11, pp. 104--110, 2014.

\bibitem{mahdavi2013energy}
H.~Mahdavi-Doost and R.~D. Yates, ``Energy {H}arvesting {R}eceivers: {F}inite
  {B}attery {C}apacity,'' in \emph{Information Theory Proceedings (ISIT), 2013
  IEEE International Symposium on}.\hskip 1em plus 0.5em minus 0.4em\relax
  IEEE, 2013, pp. 1799--1803.

\bibitem{gray2011entropy}
R.~M. Gray, \emph{Entropy and {I}nformation {T}heory}.\hskip 1em plus 0.5em
  minus 0.4em\relax Springer Science \& Business Media, 2011.

\bibitem{rajesh2014capacity}
R.~Rajesh, V.~Sharma, and P.~Viswanath, ``Capacity of {G}aussian {C}hannels
  with {E}nergy {H}arvesting and {P}rocessing {C}ost,'' \emph{Information
  Theory, IEEE Transactions on}, vol.~60, no.~5, pp. 2563--2575, 2014.

\bibitem{tse1998multiaccess}
D.~N. Tse and S.~V. Hanly, ``Multiaccess {F}ading {C}hannels. {I}. polymatroid
  {S}tructure, {O}ptimal {R}esource {A}llocation and {T}hroughput
  {C}apacities,'' \emph{Information Theory, IEEE Transactions on}, vol.~44,
  no.~7, pp. 2796--2815, 1998.

\bibitem{hanly1998multiaccess}
S.~V. Hanly and D.~N.~C. Tse, ``Multiaccess {F}ading {C}hannels.{ II}.
  delay-{L}imited {C}apacities,'' \emph{IEEE Transactions on Information
  Theory}, vol.~44, no.~7, pp. 2816--2831, 1998.

\bibitem{goyal2002optimal}
M.~Goyal, V.~Sharma, and A.~Kumar, ``Optimal {P}ower {A}llocation for
  {M}ultiaccess {F}ading {C}hannels with {M}inimum {R}ate {G}uarantees,''
  \emph{IEEE International Conference on Personal Wireless Communications},
  2002.

\bibitem{verdu2008information}
S.~Verdu and T.~Weissman, ``The {I}nformation {L}ost in {E}rasures,''
  \emph{Information Theory, IEEE Transactions on}, vol.~54, no.~11, pp.
  5030--5058, 2008.

\bibitem{walrand1988introduction}
J.~Walrand, \emph{An {I}ntroduction to {Q}ueueing {N}etworks}.\hskip 1em plus
  0.5em minus 0.4em\relax Prentice Hall, 1988.

\bibitem{barron1985strong}
A.~R. Barron, ``The {S}trong {E}rgodic {T}heorem for {D}ensities: {G}eneralized
  {S}hannon-{M}cmillan-{B}reiman {T}heorem,'' \emph{The annals of Probability},
  pp. 1292--1303, 1985.

\end{thebibliography}
 \end{document}